\documentclass[twocolumn,english,conference]{IEEEtran}
\ifCLASSINFOpdf
\else
\fi

\usepackage[T1]{fontenc}
\usepackage{babel}
\usepackage{array}
\usepackage{calc}
\usepackage{amsmath}
\usepackage{amssymb}
\usepackage{amsthm}
\usepackage{graphicx}
\usepackage{cite}
\usepackage{subfigure}
\usepackage{color}

\usepackage{enumitem}
\usepackage{multirow}
\usepackage{hhline}
\usepackage[ruled]{algorithm2e}
\usepackage{algpseudocode}

\usepackage[mathlines,switch]{lineno}

\usepackage[unicode=true,
bookmarks=true,bookmarksnumbered=true,bookmarksopen=true,bookmarksopenlevel=1,
breaklinks=false,pdfborder={0 0 0},backref=false,colorlinks=false]
{hyperref}
\hypersetup{pdftitle={Your Title},
	pdfauthor={Your Name},
	pdfpagelayout=OneColumn, pdfnewwindow=true, pdfstartview=XYZ, plainpages=false}

\makeatletter

\newtheorem{theorem}{Theorem}

\let\oldforeign@language\foreign@language
\DeclareRobustCommand{\foreign@language}[1]{%
	\lowercase{\oldforeign@language{#1}}}

\makeatother
\allowdisplaybreaks

\makeatletter
\newcommand{\thickhline}{%
	\noalign {\ifnum 0=`}\fi \hrule height 1.5pt
	\futurelet \reserved@a \@xhline
}
\newcolumntype{"}{@{\hskip\tabcolsep\vrule width 1pt\hskip\tabcolsep}}
\makeatother

\begin{document}
	%
	\title{False Data Injection Attacks on Phasor Measurements That Bypass Low-rank Decomposition}

	
	
	%
	\author{\IEEEauthorblockN{
			Jiazi Zhang,
			Zhigang Chu,
			Lalitha Sankar,
			and	Oliver Kosut}
		\IEEEauthorblockA{School of Electrical, Computer and Energy Engineering\\
			Arizona State University\\
			Tempe, AZ, 85287}}


	\maketitle
\pagestyle{plain}
	\begin{abstract}
	This paper studies the vulnerability of phasor measurement units (PMUs) to false data injection (FDI) attacks.  Prior work demonstrated that unobservable FDI attacks that can bypass traditional bad data detectors based on measurement residuals can be identified by detector based on low-rank decomposition (LD). In this work, a class of more sophisticated FDI attacks that captures the temporal correlation of PMU data is introduced. Such attacks are designed with a convex optimization problem and can always bypass the LD detector. The vulnerability of this attack model is illustrated
	on both the IEEE 24-bus RTS and the IEEE 118-bus systems.
	\end{abstract}
	

	%
	\IEEEpeerreviewmaketitle
	\global\long\def\figurename{Fig.}
	\global\long\def\tablename{TABLE}

	\section{Introduction}
	The electric power system is monitored via an extensive network of sensors in tandem with data processing algorithms, \textit{i.e.}, an intelligent cyber layer, that enables continual observation and control of the physical system. 
	In the past decade, phasor measurement units (PMUs) have been widely deployed in the electric power system to directly measure the bus voltages and phase angles. Due to its high sampling rate and accuracy, PMU has the potential to play a significant role in real-time power system state estimation (SE) \cite{Phadke_PMU1986}, dynamic security assessment \cite{Vittal2007, DSA_PMU2012}, system protection \cite{PMU_application2010}, and system awareness \cite{PMU_monitoring2010}.

	In recent years, several incidents \cite{Zeller2011,StuxnetIran,Siemens14plants,UkraineAttack} have demonstrated that the cyber layer of power system is vulnerable to cyber-attacks that impact the system operation status and lead to serious physical consequences. As increasingly important monitoring devices, PMUs are also prone to cyber-attacks. Therefore, it is crucial to evaluate the vulnerability of PMUs to potential cyber-attacks. In \cite{Lin_PMUcybersecurity2012,Beasley2014_PMU}, the authors study the cyber-security of PMU-based SE and classify the potential cyber-attacks on PMUs as communication link damage attacks, denial of service attacks, and data spoofing attacks including GPS spoofing attacks and false data injection (FDI) attacks.  
	
	In this paper, we focus on FDI attacks, wherein the attacker replaces a subset of PMU measurements with counterfeits. The most effective FDI attacks are those which are unobservable to SE. It has been established in \cite{Liu2009,Kosut2011,Hug2012,Liang2014} that FDI attacks,  when designed appropriately, can be unobservable to both DC SE and AC SE that use traditional SCADA measurements. Such attacks can be designed to specifically have a financial impact (\textit{e.g.}, on the electricity market such as in \cite{Yuan11}, \cite{Yuan2012}) or worse yet a physical impact on the system (\textit{e.g.}, \cite{Liang2015}, \cite{Zhang2016TSG}).

	To thwart unobservable FDI attacks, several protection mechanisms and attack detection approaches have been introduced recently. In \cite{Kim_PMU2013}, Kim and Tong introduce an approach to ensure system observability by placing secure PMUs so as to protect the system from FDI attacks. However, since PMU measurements can also be changed by attacker, this method cannot eliminate FDI attacks when PMUs are compromised by attackers. 
	In \cite{StatisticalFIDdetection_2013}, the authors propose a decentralized detection scheme for FDI attacks based on the Markov graph of bus phase angles. However, this method might not work well when the system experiences a disturbance.
	In \cite{Kundur_PMU2014}, Lee and Kundur introduce a detector based on Expectation-Maximization to detect FDI attacks in PMU measurements. This method needs to be solved iteratively and the convergence rates are very slow for real-time detection (\textit{e.g.,} $10^5$ iterations) for even a small test system.
	
	\textit{Related Work}: Recently, using measurements obtained from deployed PMUs in the grid, \cite{DahalPMU2012} and \cite{ChenPMU2013} illustrate the low-rank nature of PMU data. These approaches suggest that PMU measurements can be modeled as a matrix to capture both the temporal aspects (\textit{e.g.}, via the rows of the matrix) and the spatial aspects (for each time instant via the columns).

	Recently in \cite{Liu_SparseOptimization_2014,Wang_SmartGridComm2014_PMU,Gao2016_PMUIdentification}, low-rank decomposition (LD) has been proposed to detect FDI attacks on the electric power system using a block of consecutive measurement data. 
	On the other hand, the FDI attacks of most interest are those in which the attacker is not omniscient and omnipresent --- this limited knowledge and limited capabilities of FDI attacks are often captured (see, for \textit{e.g.}, \cite{Liu2009,Teixeira2010,Kosut2011,Liang2015,Zhang2016TSG,ZhangPES2016,Chu2016SmartGridComm,Zhang2017_TPS}) by restricting attacker knowledge to a subset of the network and restricting counterfeits to a small number of meters, respectively. This latter restriction along with the above mentioned low-rank properties of a block of PMU data suggests that the resulting counterfeit PMU measurement matrix can be viewed as a linear combination of a low-rank (actual) measurement matrix and a sparse attack matrix (counterfeit additions to measurement). 
	
	In \cite{Liu_SparseOptimization_2014}, the authors propose a LD approach (introduced in \cite{Xu_LRD2012} for arbitrary sparse datasets), for temporal SCADA data; specifically, they demonstrate that attacks designed without knowledge of the temporal correlations of the SCADA measurements can be detected by solving an LD problem. Furthermore, their model assumes that while the FDI attack matrix is sparse in each time instant, the attacker attacks a different set of
	measurements. While such a model is quite general, for attacks designed with a specific effect (financial or physical damage), sustaining attacks over time on the same meters can have more impact. Focusing on such sustained attacks, for PMU data, the authors of \cite{Wang_SmartGridComm2014_PMU,Gao2016_PMUIdentification} show that an LD-based detector can identify column sparse FDI
	attack matrix where the column sparsity is a result of the assumption that the attacker attacks the same set of PMU
	measurements over time. 
	
	\textit{Our Contributions}: Following \cite{Wang_SmartGridComm2014_PMU,Gao2016_PMUIdentification}  we model PMU data as a low-rank matrix. Furthermore, focusing on impactful FDI attacks, our attack model involves sustained attacks on the same meters over time, \textit{i.e.}, column sparse attacks (using the nomenclature that rows and columns indicate spatial and temporal data, respectively).	
	Although LD detector shows good performance in detecting column sparse unobservable FDI attacks on both synthetic data and some field PMU data \cite{Wang_SmartGridComm2014_PMU,Gao2016_PMUIdentification}, a question that needs to be addressed is the following: if an attacker has knowledge of the time correlation of the PMU data, can it take advantage of such knowledge and design FDI attacks that can bypass the detector? In this paper, we assume the attacker has the ability to predict the system dynamics, and we introduce a new class of FDI attacks that can bypass the LD detector. These attacks are designed with a convex optimization problem. We prove that the LD detector cannot identify the exact set of states that are modified by the attacker. We demonstrate that such attacks are unobservable for both traditional bad data detectors and the LD detector on both the IEEE 24-bus and IEEE 118-bus systems. 

\section{Preliminaries} \label{sec:SystemModel}

In this section, we introduce the models for the system, FDI attacks, and the LD detector. Throughout, we assume there are $n_b$ buses,
$n_{br}$ branches, $n_g$ generators, and $n_z$ measurements in the
system. 
\subsection{System Model}
 PMUs collect complex bus voltage and branch current measurements. The reporting rate of the PMU measurements is usually 30 times per second \cite{Phadke2009_PMUSE}. These measurements have a linear relationship with the complex bus voltage states. At each time instant $t$, the PMU measurement model can be written as
\begin{equation}
z_t=Hx_t+e_t\label{eq:ACMeasurement}
\end{equation}
where at time instant $t$, $z_t$ is the $n_z \times 1$ measurement vector; $x_t$ is the state vector of complex bus voltage; $e_t$ is an $n_{z}\times1$ noise vector assumed to be composed of independent Gaussian random variables; the complex matrix $H$ is the $n_{z}\times n_{b}$ dependency matrix between measurements and states. Note that the state can be estimated based on PMU measurements via a single weighted least squares (WLS) \cite{Phadke2009_PMUSE}, unlike traditional SCADA-based SE which requires multiple iterations due to the nonlinearity of the measurement function \cite{AburBook}.

One possible way to process PMU data is to collect over a block of time (\textit{e.g.,} 5 to 20 seconds) and then process them as a batch (see for example \cite{Gao2015_LowRankPMU}). We adopt this approach and write the PMU measurements as a matrix where each row vector corresponds to PMU measurements at one time instant and each column vector consists of the measurements collected in the same channel over a period of times. The PMU measurements in \eqref{eq:ACMeasurement} over $N$ time instants can then be collected as
\begin{equation}
Z = XH^T+E \label{eq:MeasurementMatrix}
\end{equation}   
where matrices $Z=\left[z_1^T;\: z_2^T; \hdots ;\:z_N^T\right]$, $X=\left[x_1^T;\: x_2^T;\hdots;x_N^T\right]$, and $E=\left[e_1^T;\: e_2^T;\hdots;e_N^T\right]$ are PMU measurement matrix, state matrix, and noise matrix, respectively. Note that $z_t^T$, $x_t^T$, and $e_t^T$ for $t=1, 2,\hdots,N$ are the transpose of the measurement, state, and noise column vectors, respectively, in \eqref{eq:ACMeasurement}.





\subsection{Unobservable FDI Attack Model}

Assume the attacker has control of the measurements in a subset $\mathcal{S}$ of the network, denoted as the attack subgraph. As in \cite{Liang2015}, we first distinguish between two types of buses in the network: \textit{load buses} that have load directly connected to them, and \textit{non-load buses} with no load. We assume $\mathcal{S}$ is bounded by load buses. The set of measurements in $\mathcal{S}$ are denoted as $\mathcal{J}$. In the absence of noise, an attack is defined to be unobservable if 
\begin{equation}
\bar{Z} = Z+D= Z+CH^T=\left(X+C\right)H^T+E \label{eq:AttackMatrix}
\end{equation}
where $\bar{Z}$ is the $N\times n_z$ post-attack measurement matrix, $D$ is the $N\times n_z$ attacked measurements matrix such that $D=CH^T$, and $C$ is the $N\times n_b$ attack matrix. In the following, we define the set of non-zero columns in a matrix as its column support, written as $supp\left(\cdot\right)$. Note that the attacker is constrained to inject false data only in the measurements in $\mathcal{J}$. Thus, $D$ is a column sparse matrix where $supp\left(D\right)\subseteq\mathcal{J}$. One natural way to form a column spare $D$ is to choose a column sparse $C$.

Prior work \cite{Liu2009,Teixeira2010,Kosut2011,Liang2015,Zhang2016TSG,ZhangPES2016,Chu2016SmartGridComm} considers a special case of \eqref{eq:AttackMatrix} with only one time instant, \textit{i.e.,} $N=1$. These works show that traditional bad data detectors based on measurement residuals cannot detect such FDI attacks.

\subsection{Prior Work: Attack Detection Based on Low-Rank Matrix Decomposition} 
Traditional bad data detectors based on measurement residuals cannot detect the FDI attacks introduced in \eqref{eq:AttackMatrix}. However, exploiting the low-rank nature of the high-dimensional PMU data matrix $Z$, the authors in \cite{Wang_SmartGridComm2014_PMU} propose a new attack detection mechanism based on LD so as to separate the low-rank matrix $Z$ and column sparse matrix $CH^T$ in \eqref{eq:AttackMatrix}. We now briefly review their attack assumptions and detection methodology. 

Given a measurement matrix $\bar{Z}^{\left(\text{LD}\right)}$, the measurement matrix without attack, $Z^{\left(\text{LD}\right)}$, and the attack matrix $C^{\left(\text{LD}\right)}$ can be identified by solving the following convex optimization problem:    
 \begin{flalign}
 \underset{Z^{\left(\text{LD}\right)}\in\mathbb{C}^{N\times n_z},C^{\text{LD}}\in \mathbb{C}^{N\times n_b}}{\text{minimize}} \;\;& \|Z^{\left(\text{LD}\right)}\|_* + \lambda \|C^{\left(\text{LD}\right)}\|_{1,2} \label{eq:LD_OBJ}\\
 \text{subject to} \hspace{0.6cm}& \bar{Z}^{\left(\text{LD}\right)}=Z^{\left(\text{LD}\right)}+C^{\left(\text{LD}\right)}\bar{H}^T \label{eq:LD_UnobservableAttack}
 \end{flalign}
 where $\|Z^{\left(\text{LD}\right)}\|_*$ is the nuclear norm of $Z^{\left(\text{LD}\right)}$; $\|C^{\left(\text{LD}\right)}\|_{1,2}$ is the $l_{1,2}$-norm of $C^{\left(\text{LD}\right)}$, \textit{i.e.,} the sum of $l_2$-norm of columns in $C^{\left(\text{LD}\right)}$; $\lambda$ is a weight factor; and $\bar{H}$ is the normalized dependency matrix, where for each row vector $H_i$, $\bar{H}_i=H_i / \lVert H_i\rVert$.  The objective \eqref{eq:LD_OBJ} is to minimize the rank of $Z^{*\left(\text{LD}\right)}$ (captured by its nuclear norm) and the column sparsity of $C^{*\left(\text{LD}\right)}$ (captured by its $l_{1,2}$-norm). 
 
 After obtaining the optimal solution, $\left(Z^{*\left(\text{LD}\right)},C^{*\left(\text{LD}\right)}\right)$ for \eqref{eq:LD_OBJ}--\eqref{eq:LD_UnobservableAttack}, the set of attacked measurements and states, $supp\left(C^{*\left(\text{LD}\right)}\bar{H}^T\right)$ and $supp\left(C^{\left(*\text{LD}\right)}\right)$, respectively, can be identified as the column support of $C^{*\left(\text{LD}\right)}\bar{H}^T$ and $C^{*\left(\text{LD}\right)}$. Assume there exists unobservable attacks in $\bar{Z}^{\left(\text{LD}\right)}$, such that $\bar{Z}^{\left(\text{LD}\right)}=Z+C\bar{H}^T$. The authors prove that for a specific range of $\lambda$, \textit{i.e.,} $\lambda \in \left[\lambda_\text{min},\;\lambda_\text{max}\right]$, the optimization in \eqref{eq:LD_OBJ} can successfully identify $supp(C)$, \textit{i.e.,} $supp\left(C^{*\left(\text{LD}\right)}\right)=supp\left(C\right)$, under the assumption that every nonzero
 column of $C\bar{H}^T$ does not lie in the column space of $Z$.
%
%
%
%


\section{FDI Attack Exploiting Low-Rank Property of PMU Measurement Matrix} \label{sec:LRattack}
In this section, we introduce a class of FDI attacks that cannot be detected by the LD detector in \eqref{eq:LD_OBJ}--\eqref{eq:LD_UnobservableAttack}. 
	We assume that the attacker has the following knowledge and capabilities:
		\begin{enumerate}
			\item The attacker has full system topology information.
			\item The attacker can perfectly predict the measurements in the following $N$ instances and has the capability to estimate the predicted states. 
			\item The attacker has control of the measurements in a subset $\mathcal{S}$ of the network.
		\end{enumerate}

Given a PMU measurement matrix $Z$ and the potential attacked states set $\mathcal{I}$, we propose the following optimization problem to design FDI attacks:
\begin{flalign}
\underset{C\in \mathbb{C}^{N\times n_b}}{\text{minimize}} \:\;\;& \lVert Z+C\bar{H}^T\rVert_{*} \label{eq:OBJ_AM1}\\
\text{subject to} \;\;& supp(C)\subseteq \mathcal{I} \label{eq:AM1_ColumnSparse} 
\end{flalign}
where $\lVert \cdot\rVert_*$ denotes the nuclear norm.
%
For optimal solution $C^*$, the optimal post-attack measurement matrix denoted as $\bar{Z}^*$ can be written as
\begin{equation}
\bar{Z}^*=Z+C^*\bar{H}^T. \label{eq:optAttackMatrix}
\end{equation}
The goal of the attacker is to ensure that the attacked measurement matrix $\bar{Z}^*$ is low-rank when $Z$ is low-rank. This can be approximated by minimizing the nuclear norm of $\bar{Z}^*$ as \eqref{eq:OBJ_AM1}. 
Constraint \eqref{eq:AM1_ColumnSparse} ensures that the attacker can only attack states in $\mathcal{I}$, \textit{i.e.,} $C^*$ is a column sparse matrix.

In the following, we prove that either $\bar{Z}^*$ bypasses the LD detector (\textit{i.e.,} results in $C^{*(LD)}=0$), or the LD detector identifies at least one measurement as corrupted that is not.

\begin{theorem}
Assume the attack-free measurement matrix $Z$ can bypass the LD detector, \textit{i.e.,} for $\bar{Z}^{\left(\text{LD}\right)}=Z$, $\left(Z^{*\left(\text{LD}\right)},\;C^{*\left(\text{LD}\right)}\right)=\left(Z,\mathbf{0}\right)$.
Assume the solution $C^*$ of \eqref{eq:OBJ_AM1}--\eqref{eq:AM1_ColumnSparse} is non-zero.
Then using $\bar{Z}^*$ in the LD detector, the resulting $C^{*\left(\text{LD}\right)}$ satisfies that either $C^{*\left(\text{LD}\right)}=\mathbf{0}$, or $supp\left(C^{*\left(\text{LD}\right)}\right)\not\subseteq supp\left(C^*\right)$. 
\end{theorem} 
\begin{proof}
First, we prove that for a given $Z$, $\lVert \bar{Z}^* \rVert_*\leq \lVert Z \rVert_*$.
For a given $Z$, $C=\mathbf{0}$ is always a feasible solution for \eqref{eq:OBJ_AM1}--\eqref{eq:AM1_ColumnSparse}. For $C=\mathbf{0}$, the objective $\lVert \bar{Z}\rVert_*=\lVert Z+C\bar{H}^T\rVert_*=\lVert Z\rVert_*$. Since we minimize \eqref{eq:OBJ_AM1}, the objective of $C^*$ is less than or equal to that of the feasible solution $\mathbf{0}$. That is, $\lVert \bar{Z}^*\rVert_*\leq \lVert Z \rVert_*$ always holds. 

Suppose $Z$ can bypass LD detector. That is, for input $\bar{Z}^{\left(\text{LD}\right)}=Z$,  $\left(Z^{*\left(\text{LD}\right)},\; C^{*\left(\text{LD}\right)}\right)=\left(Z,\mathbf{0}\right)$.
As we just proved
\begin{equation}
\lVert Z+C^*\bar{H}^T\rVert_*=\lVert \bar{Z}^*\rVert_*\leq \lVert Z\rVert_*=\lVert Z^{*\left(\text{LD}\right)}\rVert_*.
\end{equation} 
Thus,
\begin{equation}
\lVert Z+C^*\bar{H}^T\rVert_*\leq \lVert Z\rVert_*\leq \lVert Z\rVert_*+\lambda \lVert C^* \rVert_{1,2}.
\end{equation} 

Let $C^{*\left(\text{LD}\right)}$ be the optimal solution of the LD detector for $\bar{Z}^{\left(\text{LD}\right)}=\bar{Z}^*$. The objective \eqref{eq:LD_OBJ} for $\bar{Z}^*$ satisfies
\begin{equation}
\|\bar{Z}^*-C^{*\left(\text{LD}\right)}\bar{H}^T\|_*   + \lambda \|C^{*\left(\text{LD}\right)}\|_{1,2}\leq \lVert \bar{Z}^*\rVert_*\leq \lVert Z\rVert_*.\label{eq:proof1}
\end{equation}         	
Note that $\lVert \bar{Z}^* -C^{*\left(\text{LD}\right)}\bar{H}^T\rVert_*$ can be rewritten as $\lVert Z +\left(C^*-C^{*\left(\text{LD}\right)}\right)\bar{H}^T\rVert_*$.  

If $supp\left(C^{*\left(\text{LD}\right)}\right)\subseteq\mathcal{I}$, then 
		 \begin{equation}
		 \lVert Z +\left(C^*-C^{*\left(\text{LD}\right)}\right)\bar{H}^T\rVert_*\geq \lVert Z +C^*\bar{H}^T\rVert_*
		 \end{equation}
		 since $C^*$ is the optimal solution for \eqref{eq:OBJ_AM1}--\eqref{eq:AM1_ColumnSparse}. That is, $\bar{Z}^*$ and $C^{*\left(\text{LD}\right)}$ satisfy
		 \begin{equation}
		 \lVert \bar{Z}^*-C^{*\left(\text{LD}\right)}\bar{H}^T\rVert_*+\lambda\lVert C^{*\left(\text{LD}\right)}\rVert_{1,2} \geq \lVert \bar{Z}^*\rVert_*. \label{eq:proof2}
		 \end{equation}
		Therefore, the only solution that can satisfy both \eqref{eq:proof1} and \eqref{eq:proof2} is $C^{*\left(\text{LD}\right)}=\mathbf{0}$.
\end{proof}

 \section{Numerical Results} \label{sec:Simulation}
 In this section, we illustrate the efficacy of the unobservable FDI attacks introduced in Sec. \ref{sec:LRattack}. To this end, we first solve the attack optimization problem in \eqref{eq:OBJ_AM1}--\eqref{eq:AM1_ColumnSparse} to find the optimal attack matrix $C^*$. Subsequently, we construct the post-attack measurement matrix $\bar{Z}^*$ with $C^*$ as \eqref{eq:optAttackMatrix}. Finally, we solve the LD detection optimization problem \eqref{eq:LD_OBJ}--\eqref{eq:LD_UnobservableAttack} for $\bar{Z}^*$ to check if the attack matrix $C^*$ is detected. Throughout, we assume that the LD detector selects 2 seconds worth of PMU measurements data, \textit{i.e.,} $N=60$, while the attacker injects bad data. The test systems include the IEEE 24-bus reliability test system (RTS) and IEEE 118-bus system. The convex optimization problems for LD detection and attack design are solved with MOSEK. In the LD detection optimization problem, the weight $\lambda$ is chosen to be 1.05 for both the IEEE 24-bus and the IEEE 118-bus systems. 
 \begin{figure}[h]
 	\centering{}\includegraphics[trim=0 1cm 0 0.5cm,scale=0.38]{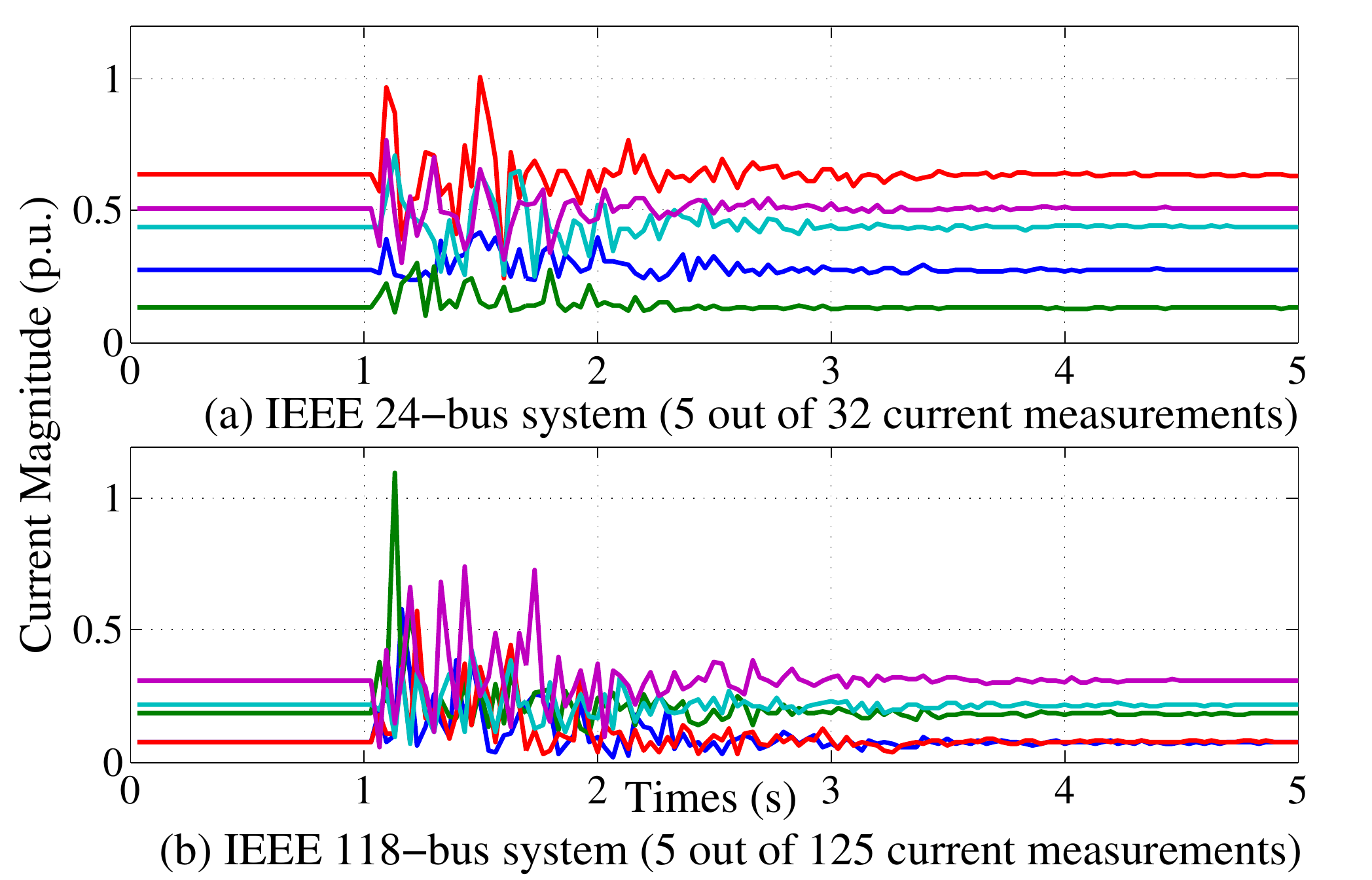}\protect\protect	\caption{Current magnitudes of the synthetic PMU data.\label{fig:LowRankMatrix}}
 \end{figure}
 \begin{figure}[h]
 	\centering{}\includegraphics[trim=0 1cm 0 1cm,scale=0.38]{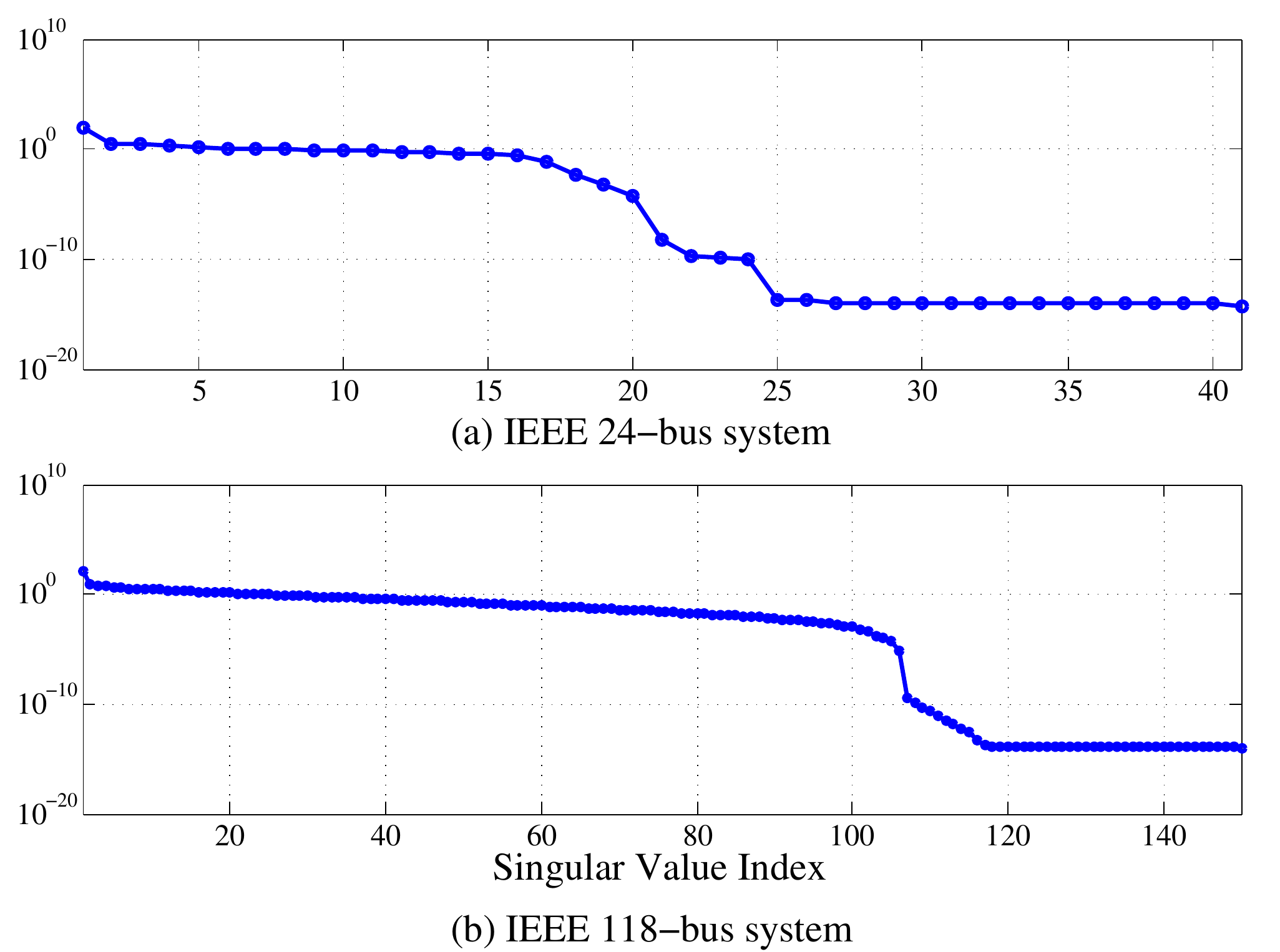}\protect\protect	\caption{Singular values of the synthetic PMU data matrix in decreasing order.\label{fig:SingularValue}}
 \end{figure}
 
 We assume both test systems are fully observable with PMU measurements. This is achieved by solving an optimal PMU placement problem as introduced in \cite{PMUplacement2004}. The details of the PMU locations and measurements for both test systems are summarized in Table \ref{tab:PMU_place}. In \cite{Wang_SmartGridComm2014_PMU}, the authors demonstrate an actual PMU dataset, which we do not have access to. Therefore, to make a fair comparison, we generate synthetic PMU data over 5 seconds in each test system. To model realistic data with a disturbance, at the first time instant $t$ after 1 second, we change the load at each bus by adding a random value $d$ to the base load, such that $d\sim \mathcal{N}\left(0,\frac{60}{1.1^{\left(t-31\right)}}\right)$. We then solve an AC power flow to obtain the measured phasors of bus voltage and branch current as measurements at time instant $t$. The resulting synthetic measurements for the IEEE 24-bus system and the IEEE 118-bus system are partially illustrated in Fig \ref{fig:LowRankMatrix}. The singular values for the synthetic measurement matrices for the IEEE 24-bus system and the IEEE 118-bus system are illustrated in Fig \ref{fig:SingularValue}. It can be seen that these synthetic measurements have the same low-rank property as the actual PMU data as illustrated in \cite{Wang_SmartGridComm2014_PMU}. The synthetic data is then broken into two parts for testing, one for $t=$1--3 seconds, the other for $t=$3--5 seconds. Observe that the measurement matrix for $t$=1--3 seconds has more disturbances than that for $t$=3--5. Furthermore, we assume noiseless measurements, \textit{i.e.}, $E=\mathbf{0}$ in \eqref{eq:AttackMatrix}.
  \begin{table}[h]
  	\renewcommand{\arraystretch}{1.3}
  	\protect\caption{Statistic results of $\lVert \bar{Z}^*\rVert_*$ in the IEEE 118-bus system.  \label{tab:Statistics}}
  	\centering
  	\begin{tabular}{c|>{\centering}p{1cm}|>{\centering}p{1cm}|>{\centering}p{1cm}|>{\centering\arraybackslash}p{0.5cm}}
  		\thickhline
  		\multirow{2}{0.8cm}{\centering Time Period} & \multicolumn{3}{c|}{$\lVert \bar{Z}^*\rVert$} & \multirow{2}{*}{$\lVert Z\rVert$}  \\
  		\hhline{~---~}
  		& Min & Max & Ave& \\
  		\hline
  		1--3 second & 116.1  & 116.7 & 116.5 & 116.8 \\
  		\hline
  		3--5 second& 56.9 & 57.1 & 57.0& 57.1\\
  		\thickhline 
  	\end{tabular}
  \end{table}
  \begin{figure}[h]
  	\centering{}\includegraphics[trim=0 0.7cm 0 1cm,scale=0.4]{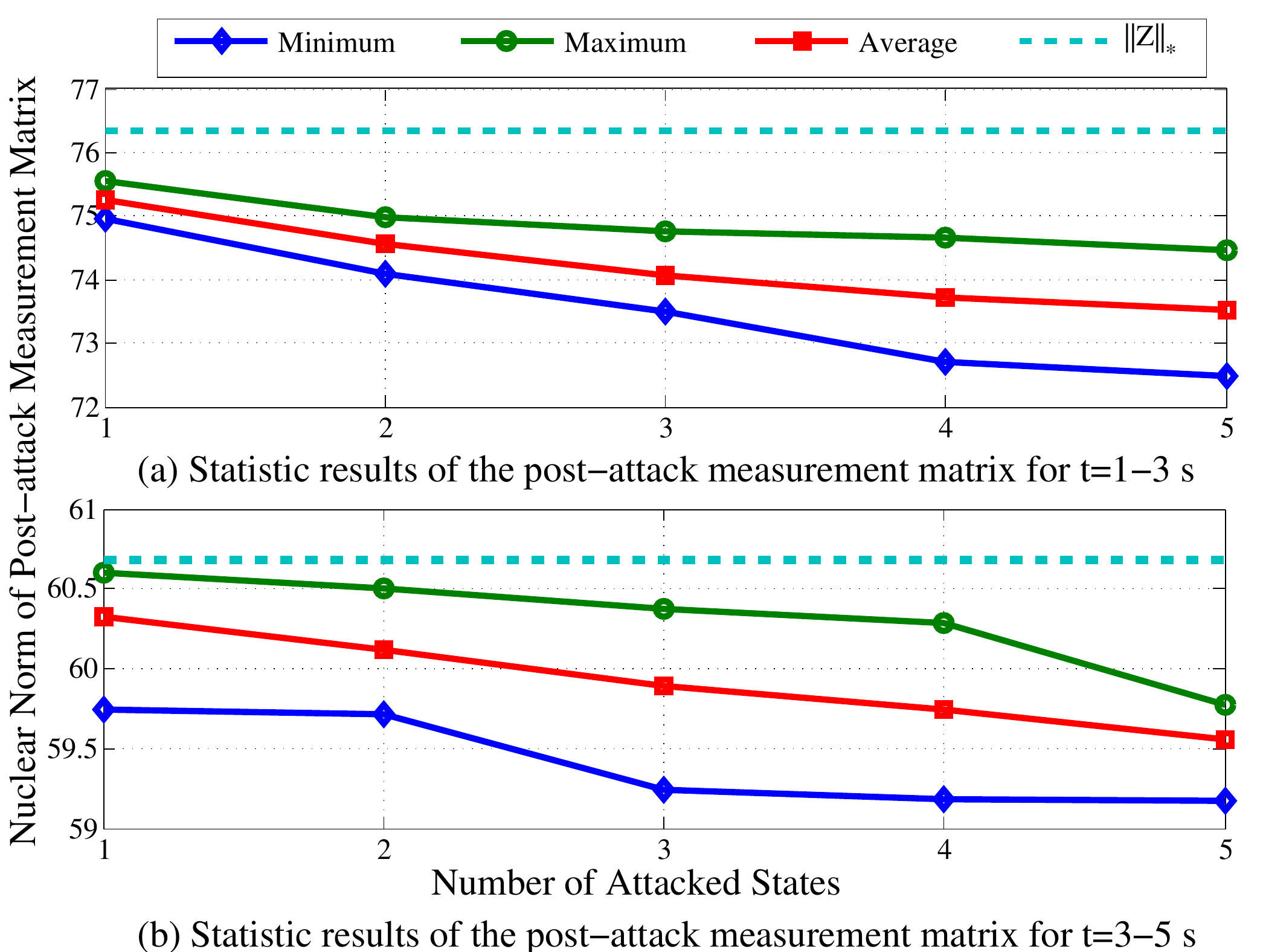}\protect\protect	\caption{Statistic results of $\lVert \bar{Z}^*\rVert_*$ in the IEEE 24-bus system. \label{fig:Statistics}}
  \end{figure}

 We exhaustively generate the unobservable attacks with all potential attacked state sets $\mathcal{I}$ for which $1\leq |\mathcal{I}| \leq 5$ in the IEEE 24-bus system; for tractability we consider only $|\mathcal{I}|=1$ in the IEEE 118-bus system. Specifically, as observed in our prior work \cite{Liang2015,Zhang2016TSG}, unobservable attacks must be designed inside a subgraph $\mathcal{S}$ which is bounded by load buses. In $\mathcal{S}$, the states of all the non-bounded buses (including load and non-load buses) have to be changed. In this paper, the attacked state sets $\mathcal{I}$ are selected according to this rule.

 \begin{figure}[h]
 	\centering{}\includegraphics[trim=0 0.7cm 0 1cm,scale=0.4]{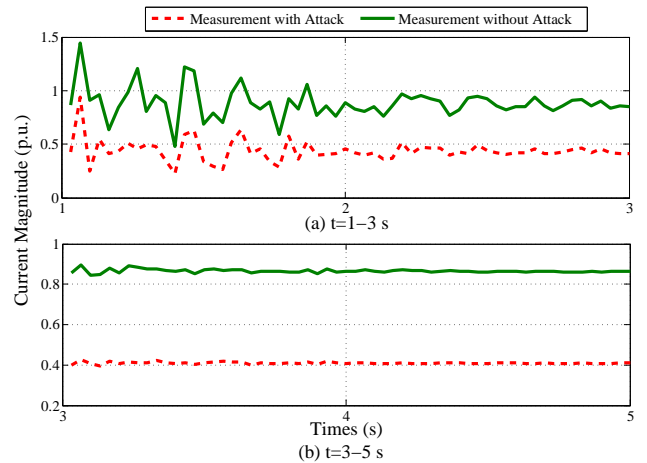}\protect\protect	\caption{Magnitude of the from side current measurement on line 12 in the IEEE 24-bus system with $\mathcal{I}=\{8\}$. \label{fig:example}}
 \end{figure}
 
   For every attack we tested, the LD detector is completely bypassed, \textit{i.e.,} $C^{*\left(\text{LD}\right)}=\mathbf{0}$.
   We summarize the statistic results including maximum, minimum, and average values of  $\lVert \bar{Z}^* \rVert_*$ for the IEEE 24-bus system and the IEEE 118-bus system in Fig. \ref{fig:Statistics} and Table \ref{tab:Statistics}, respectively.
  From these results, it can be seen that for every attack we tested, $\lVert \bar{Z}^* \rVert_* \leq \lVert Z \rVert_*$ always holds. In addition, in  Fig. \ref{fig:Statistics}, we also find that for the IEEE 24-bus system, $\lVert \bar{Z}^* \rVert_*$ gradually decreases as the number of attacked state increases.  
 
  We now illustrate a typical case in the IEEE 24-bus system in with $\mathcal{I}=\{8\}$. In Fig. \ref{fig:example}, the magnitudes of the from side current measurement on line 12 are demonstrated, before and after attack. From this case, we can see that the designed attack accurately captures the temporal correlation of the PMU measurements.

  These observations are consistent with Theorem 1. In fact, they are stronger than Theorem 1 since we did not find any case where the LD detector results in $C^{*\left(\text{LD}\right)}$ such that  $C^{*\left(\text{LD}\right)}\neq \mathbf{0}$ and  $supp\left(C^{*\left(\text{LD}\right)}\right)\not\subseteq \mathcal{I}$.
  
 \begin{table}[h]
 	\renewcommand{\arraystretch}{1.2}
 	\protect\caption{Monitored PMU measurements in both the IEEE 24-bus system and the IEEE 118-bus system.  \label{tab:PMU_place}}
 	\centering
 	\begin{tabular}{>{\centering}p{0.3cm}|>{\centering}p{1.5cm}|>{\centering\arraybackslash}p{5.8cm}}
 		\thickhline
 		\multirow{3}{*}{\rotatebox[origin=c]{90}{\parbox[c]{2.45cm}{\centering IEEE 24-bus System}}}  & Voltage (Buses with PMU) & 1, 2, 7, 9, 10, 11, 15, 17, 20 \\
 		\hhline{~--}
 		& Current (From Side) & 1, 2, 3, 4, 5, 11, 14, 15, 16, 17, 18, 19, 24, 25, 26, 27, 30, 31, 36, 37\\
 		\hhline{~--}
 		& Current (To Side) & 1, 6, 8, 9, 10, 12, 13, 14, 16, 28, 34, 35\\
    	\hline
 		\multirow{3}{*}{\centering \rotatebox[origin=c]{90}{\parbox[c]{4cm}{IEEE 118-bus System}}}  & Voltage (Buses with PMU) & 2, 5, 10, 12, 15,    17,    21,    25,    29,    34,    37,    41,    45,    49,    53,    56,    62,    64,    72,    73,    75,    77,    80,
 		85,    87,    91,    94,   101,   105,   110,   114,   116\\
 		\hhline{~--}
 		& Current (From Side) & 5,	11,	13,	17,	20,	21,	23,	26,	28,	33,	39,	40,	44,	49,	50,	52,	53,	58,	60,	62,	68,	70,	71,	74,	75,	76,	80,	82,	85,	86,	95,	97,	98,	99,	100,	101,	106,	120,	121,	123,	124,	128,	133,	135,	136,	143,	147,	148,	150,	151,	152,	153,	155,	162,	169,	170,	171,	176,	177,	178,	182,	184,	185\\
 		\hhline{~--}
 		& Current (To Side) & 1,	3,	4,	8,	9,	12,	13,	14,	15,	18,	19,	21,	22,	27,	31,	32,	35,	36,	45,	47,	48,	50,	51,	56,	61,	65,	66,	67,	68,	69,	73,	78,	79,	91,	92,	94,	111,	112,	113,	115,	116,	117,	118,	119,	120,	123,	124,	125,	127,	131,	132,	134,	140,	145,	146,	160,	166,	168,	174,	175,	180,	183\\
 		\thickhline 
 	\end{tabular}
 \end{table}
 
	\section{Concluding Remarks}
In this paper, we have studied the vulnerability of phasor measurement units to FDI attacks. Prior work demonstrated that unobservable FDI attacks that can bypass traditional bad data detectors based on measurement residuals can be identified by the LD detector. In this work, we have shown that a more sophisticated attacker that understands the temporal correlation of PMU data can exploit it to design unobservable FDI attacks that cannot be detected by the LD detector. Future work involves developing countermeasures for such attacks.

	\section*{Acknowledgment}
	This material is based upon work supported by the National Science Foundation under Grant No. CNS-1449080.
	

	%
	%

	
	
	%
	%
	%
	
	\bibliographystyle{IEEEtran}
	\bibliography{dis}

\begin{thebibliography}{10}
\providecommand{\url}[1]{#1}
\csname url@samestyle\endcsname
\providecommand{\newblock}{\relax}
\providecommand{\bibinfo}[2]{#2}
\providecommand{\BIBentrySTDinterwordspacing}{\spaceskip=0pt\relax}
\providecommand{\BIBentryALTinterwordstretchfactor}{4}
\providecommand{\BIBentryALTinterwordspacing}{\spaceskip=\fontdimen2\font plus
\BIBentryALTinterwordstretchfactor\fontdimen3\font minus
  \fontdimen4\font\relax}
\providecommand{\BIBforeignlanguage}[2]{{%
\expandafter\ifx\csname l@#1\endcsname\relax
\typeout{** WARNING: IEEEtran.bst: No hyphenation pattern has been}%
\typeout{** loaded for the language `#1'. Using the pattern for}%
\typeout{** the default language instead.}%
\else
\language=\csname l@#1\endcsname
\fi
#2}}
\providecommand{\BIBdecl}{\relax}
\BIBdecl

\bibitem{Phadke_PMU1986}
A.~G. Phadke, J.~S. Thorp, and K.~J. Karimi, ``State estimlatjon with phasor
  measurements,'' \emph{IEEE Transactions on Power Systems}, vol.~1, no.~1, pp.
  233--238, Feb 1986.

\bibitem{Vittal2007}
K.~Sun, S.~Likhate, V.~Vittal, V.~S. Kolluri, and S.~Mandal, ``An online
  dynamic security assessment scheme using phasor measurements and decision
  trees,'' \emph{IEEE Transactions on Power Systems}, vol.~22, no.~4, pp.
  1935--1943, Nov 2007.

\bibitem{DSA_PMU2012}
Y.~V. Makarov, P.~Du, S.~Lu, T.~B. Nguyen, X.~Guo, J.~W. Burns, J.~F.
  Gronquist, and M.~A. Pai, ``{PMU}-based wide-area security assessment:
  Concept, method, and implementation,'' \emph{IEEE Transactions on Smart
  Grid}, vol.~3, no.~3, pp. 1325--1332, Sept 2012.

\bibitem{PMU_application2010}
J.~D.~L. Ree, V.~Centeno, J.~S. Thorp, and A.~G. Phadke, ``Synchronized phasor
  measurement applications in power systems,'' \emph{IEEE Transactions on Smart
  Grid}, vol.~1, no.~1, pp. 20--27, June 2010.

\bibitem{PMU_monitoring2010}
Y.~Zhang, P.~Markham, T.~Xia, L.~Chen, Y.~Ye, Z.~Wu, Z.~Yuan, L.~Wang, J.~Bank,
  J.~Burgett, R.~W. Conners, and Y.~Liu, ``Wide-area frequency monitoring
  network ({FNET}) architecture and applications,'' \emph{IEEE Transactions on
  Smart Grid}, vol.~1, no.~2, pp. 159--167, Sept 2010.

\bibitem{Zeller2011}
M.~Zeller, ``Myth or reality -- does the {A}urora vulnerability pose a risk to
  my generator?'' in \emph{64th Annual Conference for Protective Relay
  Engineers}, April 2011, pp. 130--136.

\bibitem{StuxnetIran}
\BIBentryALTinterwordspacing
``The {S}tuxnet worm: A cyber-missile aimed at {I}ran,'' The Economist, Tech.
  Rep., 24 September 2010. [Online]. Available:
  \url{http://www.economist.com/blogs/babbage/2010/09/stuxnet_worm}
\BIBentrySTDinterwordspacing

\bibitem{Siemens14plants}
T.~Espiner, ``Siemens: Stuxnet infected 14 industrial plants,''
  http://www.zdnet.com/article/siemens-stuxnet-infected-14-industrial-plants/,
  September 2010.

\bibitem{UkraineAttack}
K.~Zetter, ``Inside the cunning, unprecedented hack of {U}kraine's power
  grid,''
  http://www.wired.com/2016/03/inside-cunning-unprecedented-hack-ukraines-power-grid/,
  March 2016.

\bibitem{Lin_PMUcybersecurity2012}
H.~Lin, Y.~Deng, S.~Shukla, J.~Thorp, and L.~Mili, ``Cyber security impacts on
  all-{PMU} state estimator - a case study on co-simulation platform {GECO},''
  in \emph{2012 IEEE Third International Conference on Smart Grid
  Communications (SmartGridComm)}, Nov 2012, pp. 587--592.

\bibitem{Beasley2014_PMU}
C.~Beasley, G.~K. Venayagamoorthy, and R.~Brooks, ``Cyber security evaluation
  of synchrophasors in a power system,'' in \emph{2014 Clemson University Power
  Systems Conference}, March 2014, pp. 1--5.

\bibitem{Liu2009}
Y.~Liu, P.~Ning, and M.~K. Reiter, ``False data injection attacks against state
  estimation in electric power grids,'' in \emph{Proceedings of the 16th ACM
  Conference on Computer and Communications Security}, ser. CCS '09, Chicago,
  Illinois, USA, 2009, pp. 21--32.

\bibitem{Kosut2011}
O.~Kosut, L.~Jia, R.~J. Thomas, and L.~Tong, ``{Mali}cious data attacks on the
  smart grid,'' \emph{IEEE Transactions on Smart Grid}, vol.~2, no.~4, pp.
  645--658, 2011.

\bibitem{Hug2012}
G.~Hug and J.~A. Giampapa, ``Vulnerability assessment of {AC} state estimation
  with respect to false data injection cyber-attacks,'' \emph{IEEE Transactions
  on Smart Grid}, vol.~3, no.~3, pp. 1362--1370, 2012.

\bibitem{Liang2014}
J.~Liang, O.~Kosut, and L.~Sankar, ``Cyber-attacks on {AC} state estimation:
  Unobservability and physical consequences,'' in \emph{IEEE PES General
  Meeting}, Washington, DC, July 2014.

\bibitem{Yuan11}
Y.~Yuan, Z.~Li, and K.~Ren, ``Modeling load redistribution attacks in power
  systems,'' \emph{Smart Grid, IEEE Transactions on}, vol.~2, no.~2, pp.
  382--390, June 2011.

\bibitem{Yuan2012}
------, ``Quantitative analysis of load redistribution attacks in power
  systems,'' \emph{Parallel and Distributed Systems, IEEE Transactions on},
  vol.~23, no.~9, pp. 1731--1738, Sept 2012.

\bibitem{Liang2015}
J.~Liang, L.~Sankar, and O.~Kosut, ``Vulnerability analysis and consequences of
  false data injection attack on power system state estimation,'' \emph{IEEE
  Transactions on Power Systems}, vol.~31, no.~5, pp. 3864--3872, Sept 2016.

\bibitem{Zhang2016TSG}
J.~Zhang and L.~Sankar, ``Physical system consequences of unobservable
  state-and-topology cyber-physical attacks,'' \emph{IEEE Transactions on Smart
  Grid}, vol.~7, no.~4, pp. 2016--2025, July 2016.

\bibitem{Kim_PMU2013}
J.~Kim and L.~Tong, ``On phasor measurement unit placement against state and
  topology attacks,'' in \emph{2013 IEEE International Conference on Smart Grid
  Communications (SmartGridComm)}, Oct 2013, pp. 396--401.

\bibitem{StatisticalFIDdetection_2013}
H.~Sedghi and E.~Jonckheere, ``Statistical structure learning of smart grid for
  detection of false data injection,'' in \emph{2013 IEEE Power Energy Society
  General Meeting}, July 2013, pp. 1--5.

\bibitem{Kundur_PMU2014}
D.~Lee and D.~Kundur, ``Cyber attack detection in {PMU} measurements via the
  expectation-maximization algorithm,'' in \emph{2014 IEEE Global Conference on
  Signal and Information Processing (GlobalSIP)}, Dec 2014, pp. 223--227.

\bibitem{DahalPMU2012}
N.~Dahal, R.~L. King, and V.~Madani, ``Online dimension reduction of
  synchrophasor data,'' in \emph{PES T{\&}D 2012}, May 2012, pp. 1--7.

\bibitem{ChenPMU2013}
Y.~Chen, L.~Xie, and P.~R. Kumar, ``Dimensionality reduction and early event
  detection using online synchrophasor data,'' in \emph{2013 IEEE Power Energy
  Society General Meeting}, July 2013, pp. 1--5.

\bibitem{Liu_SparseOptimization_2014}
L.~Liu, M.~Esmalifalak, Q.~Ding, V.~A. Emesih, and Z.~Han, ``Detecting false
  data injection attacks on power grid by sparse optimization,'' \emph{IEEE
  Transactions on Smart Grid}, vol.~5, no.~2, pp. 612--621, March 2014.

\bibitem{Wang_SmartGridComm2014_PMU}
M.~Wang, P.~Gao, S.~G. Ghiocel, J.~H. Chow, B.~Fardanesh, G.~Stefopoulos, and
  M.~P. Razanousky, ``Identification of {"unobservable"} cyber data attacks on
  power grids,'' in \emph{2014 IEEE International Conference on Smart Grid
  Communications (SmartGridComm)}, Nov 2014, pp. 830--835.

\bibitem{Gao2016_PMUIdentification}
P.~Gao, M.~Wang, J.~H. Chow, S.~G. Ghiocel, B.~Fardanesh, G.~Stefopoulos, and
  M.~P. Razanousky, ``Identification of successive {"}unobservable {"} cyber
  data attacks in power systems through matrix decomposition,'' \emph{IEEE
  Transactions on Signal Processing}, vol.~64, no.~21, pp. 5557--5570, Nov
  2016.

\bibitem{Teixeira2010}
A.~Teixeira, S.~Amin, H.~Sandberg, K.~Johansson, and S.~Sastry, ``Cyber
  security analysis of state estimators in electric power systems,'' in
  \emph{Decision and Control (CDC), 2010 49th IEEE Conference on}, Dec 2010,
  pp. 5991--5998.

\bibitem{ZhangPES2016}
J.~Zhang, Z.~Chu, L.~Sankar, and O.~Kosut, ``False data injection attacks on
  power system state estimation with limited information,'' in \emph{2016 IEEE
  Power and Energy Society General Meeting (PESGM)}, July 2016, pp. 1--5.

\bibitem{Chu2016SmartGridComm}
Z.~Chu, J.~Zhang, O.~Kosut, and L.~Sankar, ``Evaluating power system
  vulnerability to false data injection attacks via scalable optimization,'' in
  \emph{2016 IEEE International Conference on Smart Grid Communications
  (SmartGridComm)}, Nov 2016, pp. 260--265.

\bibitem{Zhang2017_TPS}
\BIBentryALTinterwordspacing
J.~Zhang, Z.~Chu, L.~Sankar, and O.~Kosut, ``Can attackers with limited
  information exploit historical data to mount successful false data injection
  attacks on power systems?'' \emph{IEEE Transactions on Power Systems}, 2017,
  submitted. [Online]. Available: \url{https://arxiv.org/pdf/1703.07500.pdf}
\BIBentrySTDinterwordspacing

\bibitem{Xu_LRD2012}
H.~Xu, C.~Caramanis, and S.~Sanghavi, ``Robust {PCA} via outlier pursuit,''
  \emph{IEEE Transactions on Information Theory}, vol.~58, no.~5, pp.
  3047--3064, May 2012.

\bibitem{Phadke2009_PMUSE}
A.~G. Phadke, J.~S. Thorp, R.~F. Nuqui, and M.~Zhou, ``Recent developments in
  state estimation with phasor measurements,'' in \emph{2009 IEEE/PES Power
  Systems Conference and Exposition}, March 2009, pp. 1--7.

\bibitem{AburBook}
A.~Abur and A.~G. Exposito, \emph{Power System State Estimation: Theory and
  Implementation}.\hskip 1em plus 0.5em minus 0.4em\relax New York: CRC Press,
  2004.

\bibitem{Gao2015_LowRankPMU}
M.~Wang, J.~H. Chow, P.~Gao, X.~T. Jiang, Y.~Xia, S.~G. Ghiocel, B.~Fardanesh,
  G.~Stefopolous, Y.~Kokai, N.~Saito, and M.~Razanousky, ``A low-rank matrix
  approach for the analysis of large amounts of power system synchrophasor
  data,'' in \emph{2015 48th Hawaii International Conference on System
  Sciences}, Jan 2015, pp. 2637--2644.

\bibitem{PMUplacement2004}
B.~Xu and A.~Abur, ``Observability analysis and measurement placement for
  systems with {PMU}s,'' in \emph{IEEE PES Power Systems Conference and
  Exposition, 2004.}, Oct 2004, pp. 943--946 vol.2.

\end{thebibliography}

\end{document}